\documentclass[sigconf]{acmart}

\usepackage{braket}
\usepackage{cleveref}
\usepackage{listings}
\usepackage{algorithm2e}
\RestyleAlgo{ruled}
\crefname{algocf}{alg.}{algs.}
\Crefname{algocf}{Algorithm}{Algorithms}
\usepackage{hyperref}
\usepackage{tikz} 
\usetikzlibrary{positioning} 
\usetikzlibrary{quantikz2} 
\usepackage[noend]{algpseudocode} 
\usepackage{orcidlink} 
\theoremstyle{definition}
\newtheorem{definition}{Definition}[section]
\newtheorem{theorem}{Theorem}[section]
\newtheorem{lemma}{Lemma}[section]

\newtheorem*{remark}{Remark}
\newtheorem{example}{Example}[section]

\newcommand{\vcwdouble}[3]{
    \arrow[from=#1,to=#2,arrows, Rightarrow, double distance=.25mm, line width=.3mm, ->, shorten >=1pt, >=latex, "#3" {anchor=south west, yshift=1pt,xshift=2pt},  at end ] 
}

\AtBeginDocument{%
  }

\setcopyright{acmlicensed}
\copyrightyear{2025}
\acmYear{2025}
\setcopyright{cc}
\setcctype{by}
\acmConference[CF '25]{22nd ACM International Conference on Computing Frontiers}{May 28--30, 2025}{Cagliari, Italy}
\acmBooktitle{22nd ACM International Conference on Computing Frontiers (CF '25), May 28--30, 2025, Cagliari, Italy}\acmDOI{10.1145/3719276.3725187}
\acmISBN{979-8-4007-1528-0/2025/05}

\begin{document}

\title{Unleashing Optimizations in Dynamic Circuits through Branch Expansion}

\author{Yanbin Chen}

\email{yanbin.chen@tum.de}
\orcid{0000-0002-1123-1432}
\affiliation{%
  \institution{Technical University of Munich}
  \city{Munich}
  \country{Germany}
}

\renewcommand{\shortauthors}{Yanbin Chen.}

\begin{abstract}
    Dynamic quantum circuits enable adaptive operations through intermediate measurements and classical feedback. Current transpilation toolchains, such as Qiskit and T$\ket{\text{ket}}$, however, fail to fully exploit branch-specific simplifications. In this work, we propose recursive branch expansion as a novel technique which systematically expands and refines conditional branches. Our method complements existing transpilers by creating additional opportunities for branch-specific simplifications without altering the overall circuit functionality. Using randomly generated circuits with varying patterns and scales, we demonstrate that our method consistently reduces the depth and gate count of execution paths of dynamic circuits. We also showcase the potential of our method to enable optimizations on error-corrected circuits.
\end{abstract}



\maketitle

\section{Introduction}\label{sec:intro}
Dynamic circuits or circuits with conditional branching are circuits with the ability to adapt operations based on intermediate measurement outcomes or classical control signals \cite{ibm_mid_circ_meas_available_2021, nation_ibm_howtomidcircmeas_2021, ibm_dynamic_circuit}.
This flexibility significantly enhances the expressiveness of quantum algorithms, allowing for adaptive strategies that can address challenges in error correction, optimization, and complex computational tasks \cite{briegel2009measurement, roffe2019quantum, lidar2013quantum, fedorov2022vqe, doi:10.1137/S0097539795293172, dong2022ground}. 
Numerous transpilation passes have been developed to optimize quantum circuits, and the continued evolution of tools like Qiskit, T$\ket{\text{Ket}}$, and PyZX is expected to introduce even more optimization techniques in the future \cite{Qiskit, Sivarajah_tket_2021, EPTCS318.14}.
However, we believe that the potential of existing transpilation passes for optimizing dynamic circuits has not been fully exploited, as illustrated in the following examples.

\begin{example}\label{ex:motivating-example-1}
Consider the quantum program in \cref{fig:ex-single-qubit-circuit}. This program constructs a single-qubit circuit, where an $X$ gate and a $Y$ gate are first applied, and then an expression \texttt{runtime\_expr} is evaluated at runtime: If \texttt{runtime\_expr} evaluates to $True$, then the $6$ gates in the if-branch are subsequently applied; otherwise the $6$ gates are not applied. In the end, a $Y$ gate and a $Z$ gate are applied. If we run the transpiler of Qiskit, one of the state-of-art quantum compilation toolchains, on the program \cref{fig:ex-single-qubit-circuit}, in particular with the specified native Clifford gate set $\{X, Y, Z, S, H, CX\}$ and with optimization level being $3$, which gives the most aggressive optimization to minimize gate count and circuit depth, we get an optimized program \cref{fig:qiskit-transpile-on-ex-circuit}. We observe that Qiskit only optimizes the circuit by canceling out the two Hadamard gates in the if-branch, but nothing more. However, if we re-structure the program into \cref{fig:partical-elim-on-ex-circuit}, one may easily notice that when \texttt{runtime\_expr} evaluates to $True$, the whole program is equivalent to an identity transformation, because every pair of $X$ gates, $Y$ gates, $Z$ gates, and Hadamard gates cancels out, and when \texttt{runtime\_expr} evaluates to $False$, the two $Y$ gates cancel out with each other.
Therefore, we can optimize the program \cref{fig:ex-single-qubit-circuit} into a highly simplified version \cref{fig:partical-elim-on-ex-circuit-further}.
Comparing the optimized program \cref{fig:qiskit-transpile-on-ex-circuit} by Qiskit with the optimized program \cref{fig:partical-elim-on-ex-circuit-further} by our analysis, we see that extra optimizations are achieved, because we take into account opportunities of circuit simplification that are only available when taking one of the conditional branches.
\end{example}

\begin{example}
The same experiment as in \cref{ex:motivating-example-1} is performed using T$\ket{\text{Ket}}$’s compilation passes and the same limitation is observed. The optimization passes in T$\ket{\text{Ket}}$ do not merge operations across the classically controlled blocks and the surrounding circuit, limiting its ability to simplify the program. This further highlights the need for optimization techniques that consider the interaction between conditional branches and the rest of the circuit.
\end{example}

The goal of this paper is to explore and demonstrate how existing transpilation techniques can be further extended to effectively optimize dynamic quantum circuits by using branch-specific simplifications. We aim to identify optimization opportunities that are currently underutilized in state-of-the-art quantum compilation toolchains, as highlighted in the example above.

\begin{figure}
    \centering
    \begin{verbatim}
    qc = quantum_circuit(qubit_num = 1)
    qc.x(0)
    qc.y(0)
    if(runtime_expr){
        qc.y(0)
        qc.x(0)
        qc.h(0)
        qc.h(0)
        qc.z(0)
        qc.y(0)
    }
    qc.y(0)
    qc.z(0)
    \end{verbatim}
    \caption{A circuit with one if-condition where the value of expression \texttt{runtim\_expr} can only be determined at runtime.}
    \label{fig:ex-single-qubit-circuit}
\end{figure}

\begin{figure}
    \centering
        \begin{verbatim}
    qc = quantum_circuit(qubit_num = 1)
    qc.x(0)
    qc.y(0)
    if(runtime_expr){
        qc.y(0)
        qc.x(0)
        qc.z(0)
        qc.y(0)
    }
    qc.y(0)
    qc.z(0)
    \end{verbatim}
    \caption{Resulting circuit by running Qiskit's transpiler on the program \cref{fig:ex-single-qubit-circuit}. Specifically, we run the command \texttt{transpile(qc, basis\_gates={'x', 'y', 'z', 's', 'h', 'cx'}, optimization\_level=3)}.}
    \label{fig:qiskit-transpile-on-ex-circuit}
\end{figure}

\begin{figure}
    \centering
    \begin{verbatim}
    qc = quantum_circuit(qubit_num = 1)
    if(runtime_expr){
        qc.x(0)
        qc.y(0)
        qc.y(0)
        qc.x(0)
        qc.h(0)
        qc.h(0)
        qc.z(0)
        qc.y(0)
        qc.y(0)
        qc.z(0)
    }
    else{
        qc.x(0)
        qc.y(0)
        qc.y(0)
        qc.z(0)
    }
    \end{verbatim}
    \caption{Re-structure of the program \cref{fig:ex-single-qubit-circuit}.}
    \label{fig:partical-elim-on-ex-circuit}
\end{figure}

\begin{figure}
    \centering
    \begin{verbatim}
    qc = quantum_circuit(qubit_num = 1)
    if(runtime_expr){
        
    }
    else{
        qc.x(0)
        qc.z(0)
    }
    \end{verbatim}
    \caption{The optimized quantum program equivalent to the program \cref{fig:ex-single-qubit-circuit}. It is achieved by canceling out each pair of Pauli gates and Hadamard gates in both branches in the program \cref{fig:partical-elim-on-ex-circuit}.}
    \label{fig:partical-elim-on-ex-circuit-further}
\end{figure}

\section{Preliminaries}\label{sec:prelim}
In this section, we focus solely on introducing the concepts and notations used in this paper. We avoid delving into a broader review of quantum computing fundamentals, which are assumed to be familiar to the reader. For readers seeking a comprehensive introduction to the principles and techniques of quantum computing, we recommend referring to the following literature: \cite{nielsen_QC_2012, rieffel2000introduction, kaye2006introduction}.

\paragraph{Sequential composition}For two circuits $C_1$ and $C_2$ acting on the same set of qubits, we use $C_2\cdot C_1$ or $C_2 C_1$ to denote the sequential composition of them, where $C_2$ is applied immediately after $C_1$.

\paragraph{Conditionals}Conditionals are essential constructs in quantum programs, enabling dynamic control flow based on runtime information. By testing a condition, a conditional decides between executing one of two sub-circuits. This mechanism is crucial for algorithms involving feedback, error correction, or adaptive operations, such as quantum error correction. 
In this paper, we use the following notation to denote conditionals in circuits.
\begin{definition}[Conditional]
    A conditional $cdt$ has two branches $cdt.i$ and $cdt.e$ and one condition $cdt.c$, as shown in \cref{fig:ex-cdt-text}.
    We use the graphical notation \cref{fig:ex-cdt} to express conditionals in circuit diagrams.
    When executing $cdt$, its condition $cdt.c$ is tested first: If $cdt.c$ evaluates to $True$, then sub-circuit $cdt.i$ gets subsequently executed; otherwise sub-circuit $cdt.e$ gets subsequently executed.
    \label{def:conditional}
\end{definition}

\begin{figure}
    \centering
    \begin{verbatim}
    if(cdt.c){
        cdt.i
    }
    else{
        cdt.e
    }   
    \end{verbatim}
    \caption{A typical conditional \texttt{cdt} in the quantum program.}
    \label{fig:ex-cdt-text}
\end{figure}

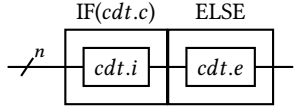
\begin{figure}
    \centering
    \begin{quantikz}
    &\qwbundle{n}&\gate{cdt.i}\gategroup[1,steps=1,style={inner
     sep=3.7pt}]{IF($cdt.c$)}&\gate{cdt.e}\gategroup[1,steps=1,style={inner
     sep=3.7pt}]{ELSE}&
    \end{quantikz}
    \caption{The graphical notation of the conditional $cdt$. }
    \label{fig:ex-cdt}
\end{figure}
\begin{remark}
    A conditional is allowed to have an if-branch without an else-branch. In this case, it has an implicit else-branch that only contains an identity transformation, and we omit its else-branch in our notation, as exemplified in the program \cref{fig:ex-single-qubit-circuit}.
\end{remark}

In dynamic circuits, conditionals can contain other conditionals as part of their branches, which allows for more complex execution paths where the evaluation of one conditional may lead to another conditional being evaluated within its branches. \Cref{def:nested-conditional} extends \cref{def:conditional} to account for this possibility.
\begin{definition}[Nested and nesting conditional] A nested conditional is a conditional that resides within the branch of another conditional. A nesting conditional is a conditional of which one branch contains another conditional.
\label{def:nested-conditional}
\end{definition}

The nesting depth of a conditional is determined by the most deeply nested conditional found within either the $True$- or $False$-branch of the conditional, as formally defined in \cref{def:nesting-depth}.

\begin{definition}[Nesting depth of conditionals]  For a conditional $cdt$ with branches $cdt.i$ and $cdt.e$, the nesting depth, denoted as $n\text{-}depth(cdt)$, is calculated as follows:

\[
\begin{cases} 
1, \qquad\qquad \text{if } cdt.i \text{ and } cdt.e \text{ contain no conditionals} \\
1 + \max\left(n\text{-}depth(cdt.i), n\text{-}depth(cdt.e) \right), \ \text{otherwise}
\end{cases}
\]
\label{def:nesting-depth}
\end{definition}

\Cref{def:size-nested-conditional} introduces the size of a conditional, which provides a measure of the structural complexity of a conditional based on the depth of its nesting. The size of a conditional is defined as the total number of execution paths or distinct branches it contains when accounting for all possible outcomes of its nesting conditionals.

\begin{definition}[Size of conditionals]  For a conditional with nesting depth $n$ (as defined in \cref{def:nesting-depth}), its size is $2^n$.
\label{def:size-nested-conditional}
\end{definition}

\begin{example}
    \cref{fig:ex-nested-conditional} illustrates a nesting conditional with nesting depth $2$ and size $2^2 = 4$.

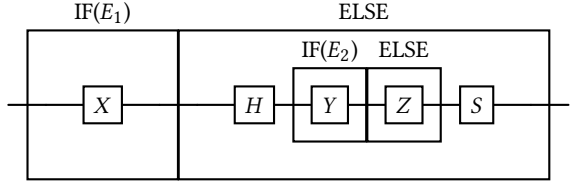
\begin{figure}
    \centering
    \begin{quantikz}
  &&\gate{X}\gategroup[1,steps=1,style={inner sep=18pt}]{IF($E_1$)}&&&\gate{H}\gategroup[1,steps=4,style={inner sep=18pt}]{ELSE}& \gate{Y}\gategroup[1,steps=1,style={inner sep=3.7pt}]{IF($E_2$)}&\gate{Z}\gategroup[1,steps=1,style={inner sep=3.7pt}]{ELSE}&\gate{S}&&
    \end{quantikz}
    \caption{A nesting conditional with nesting depth $2$, where $E_1$ and $E_2$ are expressions.}
    \label{fig:ex-nested-conditional}
\end{figure}
\end{example}

When working with quantum circuits that include conditionals, it becomes essential to evaluate their cost not just as a static structure, but in terms of the actual operations executed during runtime. This leads us to the concept of an execution path, as given in \cref{def:execution-path}, which captures the sequence of operations that are executed as a result of evaluating the conditional expressions in the circuit.

\begin{definition}[Execution path]
    An execution path of a quantum circuit is a specific sequence of operations that are executed when running the circuit, determined by the evaluation of all conditionals. 
    \label{def:execution-path}
\end{definition}

\begin{example}
    When running the program \cref{fig:qiskit-transpile-on-ex-circuit}, if $\texttt{runtime\_expr}$ evaluates to $False$, then the execution path will be $(\texttt{qc.x(0)},$ $ \texttt{qc.y(0)}, $ $\texttt{qc.y(0)}, \texttt{qc.z(0)})$; if $\texttt{runtime\_expr}$ evaluates to $True$, then the execution path will be $(\texttt{qc.x(0)},$ $\texttt{qc.y(0)},$ $\texttt{qc.y(0)},$ $\texttt{qc.x(0)},$ $\texttt{qc.z(0)}, $ $\texttt{qc.y(0)}, \texttt{qc.y(0)}, \texttt{qc.z(0)})$.
\end{example}

\begin{definition}[Size of program]
    Given a program $\mathcal{P}$, we denote its size by $|\mathcal{P}|$, which is given by the number of instructions in $\mathcal{P}$. 
\end{definition}

\begin{example}
    The size of program \cref{fig:ex-single-qubit-circuit} is $3 + 1 + 6 + 2 = 12$. In particular, there are $3$ instructions before the if-branch, $1$ instruction to check the $\texttt{runtime\_expr}$, $6$ instructions inside the if-branch, and $2$ instructions after the if-branch.
    The size of program \cref{fig:partical-elim-on-ex-circuit} is $1+1+10+4=16$, because there are $1$ instruction before the if-else-branch, $1$ instruction to test condition, $10$ instructions in the if-branch, and $4$ instructions in the else-branch. The size of the sub-program \cref{fig:ex-nested-conditional} is $1+1+2+1+1+1=7$, since in the outer conditional, there are $1$ instruction to test $E_1$, $1$ instruction in the if-branch, and $2$ instructions plus one conditional in the else-branch; in the inner conditional, there are $1$ instruction to test $E_2$, $1$ instruction in its if-branch, and $1$ instruction in the else-branch.
\end{example}

\section{Method}\label{sec:methods}
The core idea of our method is to bring sub-circuits that are outside of conditionals into each branch of the conditionals. By doing so, we aim to expose more opportunities for optimization. However, this approach is subject to a critical constraint: we cannot move a sub-circuit into the body of a conditional if the expression being tested by the conditional depends on that sub-circuit. So, it is necessary to formalize the notion of dependency between expressions and sub-circuits. Specifically, we must determine when an expression relies on the execution of certain sub-circuits to ensure correctness in our transformations. 
In \cref{def:dependency} and \cref{def:irr-dependency}, our definitions for the dependency of expressions on sub-circuits are presented.

\begin{definition}[Dependence of expressions on circuits]
    An expression $expr$ depends on a sub-circuit $C$, if and only if the evaluation of $expr$ is possible only after executing $C$.
    \label{def:dependency}
\end{definition}

\begin{definition}[Irreducible dependency]
    If an expression $expr$ depends on a sub-circuit $C$, and there exist \textbf{no} sub-circuits $C_a$ and $C_b$ such that $C = C_bC_a$ and $expr$ does not depend on $C_b$, then $expr$ irreducibly depends on $C$.
    \label{def:irr-dependency}
\end{definition}

\begin{example}
In the circuit \cref{fig:ex-expr-depend-on-circuit}, the expression $c[0]$ depends on the sub-circuit $C_1$, because the measurement in $C_1$ needs to be performed to determine the correct value in the classical register $c$; whereas $c[0]$ does not depend on the sub-circuit $C_2$. Additionally, the expression $c[0]$ depends on $C_3$, but it does not irreducibly depends on $C_3$, because $C_3 = C_2C_1$ and the $c[0]$ does not depend on $C_2$; the expression $c[0]$, however, does irreducibly depend on $C_1$.
\begin{figure}
    \centering
    \begin{quantikz}
     &&\gate{H}\gategroup[2,steps=2,style={dashed}]{$C_1$}\gategroup[2,steps=4,style={inner sep=17pt, dashed},label style={label position=below,anchor=north,yshift=-0.2cm}]{$C_3$}&\meter{}\vcwdouble{1-4}{2-4}{}&&\gate{Z}\gategroup[2,steps=1,style={dashed}]{$C_2$}&&&\gate{X}\gategroup[1,steps=1,style={inner
     sep=3.7pt}]{IF($c[0]$)}& \gate{Y}\gategroup[1,steps=1,style={inner
     sep=3.7pt}]{ELSE} & \\ \setwiretype{c}
    \lstick{$c$}&&&&&&&&&&
    \end{quantikz}
    \caption{Expression $c[0]$ depends on a $C_1$.}
    \label{fig:ex-expr-depend-on-circuit}
\end{figure}
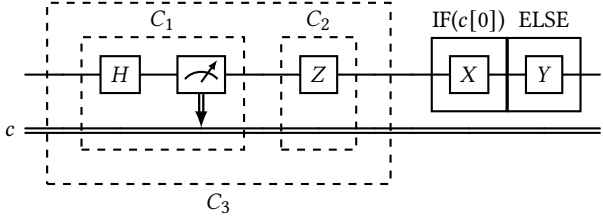
\end{example}

With the notions of dependency and irreducible dependency established, we are now equipped to address one of the key transformations in our method: branch expansion. The goal of branch expansion is to bring sub-circuits that are independent of the condition into both branches of a conditional. By doing so, we expose new opportunities for circuit optimization within each branch.
\Cref{theorem:branch-expansion} formally explains how branch expansion works under specific dependency constraints.

\begin{theorem}[Branch expansion]
For sub-circuits $C_0$, $C_1$, $C_2$, $C_{i}$, $C_{e}$ acting on the same set of qubits, and an expression $E$ that irreducibly depends on $C_0$ and does not depend on $C_1$ by \cref{def:dependency} and \cref{def:irr-dependency}, the following equivalence holds.
\begin{align*}
     &\begin{quantikz}
     &\qwbundle{n}&\gate{C_0}& \gate{C_1} &\gate{C_{i}}\gategroup[1,steps=1,style={inner
     sep=3.7pt}]{IF($E$)} &\gate{C_{e}}\gategroup[1,steps=1,style={inner
     sep=3.7pt}]{ELSE}&\gate{C_2}&
     \end{quantikz}
     \qquad =  \\
      &\begin{quantikz}[column sep = 0.4cm]
     &\qwbundle{n} &\gate{C_0}&\gate{C_1}\gategroup[1,steps=3,style={inner
     sep=2.2pt}]{IF($E$)} &\gate{C_{i}}&\gate{C_2}&\gate{C_1}\gategroup[1,steps=3,style={inner
     sep=2.2pt}]{ELSE}&\gate{C_{e}}&\gate{C_2}&
     \end{quantikz}
\end{align*}
\label{theorem:branch-expansion}
\end{theorem}
\begin{proof}
It follows directly the definition of the conditional.
\end{proof}

Since the basic idea of our method is to use \cref{theorem:branch-expansion} to maximally expand each conditional in the program, we design and implement a procedure, $rec\_branch\_expand$, of which the pseudo-code is shown in \cref{alg:optimization}. Crucially, the procedure runs recursively: if a conditional is nested within another, we proceed to expand every conditional within both branches of the outer conditional after expanding it. 
The key principle guiding the design of this procedure is to prioritize the expansion of conditionals at the outermost layer before expanding those nested deeper. 
This is because a nested conditional can only be expanded within the scope of its enclosing conditional—it cannot expand to include any code that lies outside the outer conditional. By expanding the outer conditionals first, we effectively maximize the range of code available for the subsequent expansion of inner conditionals, ensuring that their expansion potential is fully realized.

\paragraph{Controlling growth of program size with $depth\_limit$}
The parameter $depth\_limit$ is introduced to prevent uncontrolled growth in the size of the program during the branch expansion process: When a conditional is expanded, it is possible for other conditionals that were previously outside of the expanded conditional to become nested within it. This increases the nesting depth of these conditionals by \cref{def:nesting-depth}. Since the size of a conditional grows exponentially with its nesting depth by \cref{def:size-nested-conditional}, recursively expanding each conditional from outer layers to inner layers can lead to an exponential increase in program size. To mitigate this, we enforce an upper limit on the maximum nesting depth through the $depth\_limit$ parameter to ensure that the size of conditionals remains manageable in the expanded program.

\begin{algorithm}
\caption{Recursive maximum branch expansion}
\label{alg:optimization}
\KwData{$C \in dynamic \ circuits$, $depth\_limit$}
\KwResult{$rec\_branch\_expand(C,\ depth\_limit)$}
$C_{expanded} \gets C$, 
$d \gets depth\_limit$ \;
\For{$conditional \ cdt \in C_{expanded}$}{
    represent $C$ as $C_a \cdot cdt \cdot C_b$,\\ 
    where $cdt.c = expr$\;  
    \If{$expr$ does not depend on $C_a$}{
        $cdt.i \gets C_a \cdot (cdt.i) \cdot C_b$\;
        $cdt.e \gets C_a \cdot (cdt.e) \cdot C_b$\;
        
    }
    \Else{
        find sub-circuits $C_0$ and $C_1$, such that:\\
        \quad(a) $C_a = C_1\cdot C_0$, \\
        \quad(b) $expr$ irreducibly depends on $C_0$,\\
        \quad(c) $expr$ does not depend on $C_1$\;
        by \cref{theorem:branch-expansion}: \\
        \quad$cdt.i \gets C_1 \cdot (cdt.i) \cdot C_b$\;
        \quad$cdt.e \gets C_1 \cdot (cdt.e) \cdot C_b$\;
        \If{$d > 0$}{
            $cdt.i \gets rec\_branch\_expand (cdt.i, d-1) $\;
            $cdt.e \gets rec\_branch\_expand (cdt.e, d-1) $\;
        }
    }
}
\Return $C_{expanded}$ \;
\end{algorithm}

\begin{remark}
It is important to note that \cref{alg:optimization} is not an optimization pass in itself but rather a preprocessing step that enhances the effectiveness of existing optimization frameworks. By exposing hidden or latent opportunities for optimization in the circuit, our technique operates orthogonally to existing frameworks. This means that any optimization framework, regardless of its design, can benefit from incorporating our method as a preliminary step, ultimately leading to more effective circuit optimizations.
\end{remark}

\paragraph{Improved workflow of optimization}
With the procedure \cref{alg:optimization} explained, we can now conclude the explanation of how our method integrates into the broader optimization workflow. The approach is straightforward: consider any existing circuit optimization framework, such as the transpiler in Qiskit. By running \cref{alg:optimization} first, we prepare the circuit by expanding branches and exposing new opportunities for optimization. Once this preprocessing step is complete, we apply the chosen circuit optimizer to the expanded circuit.

\paragraph{Trade-off and asymptotic analysis} When applying our method \cref{alg:optimization} to a program, a key side effect is an increase in the program's size due to the expansion of conditionals. 
Since the runtime of any circuit transpiler, such as Qiskit and t$\ket{\text{ket}}$, depends on the size of the input program, this increase in program size leads to additional computational overhead.
The underlying principle here is a trade-off between runtime and optimization potential: by increasing $depth\_limit$ and enabling our method to expand deeper nested conditionals, we expose more opportunities for optimization that might otherwise remain hidden. 
The asymptotic bounds in \cref{theorem:program-size-growth} and \cref{theorem:time-complexity} quantify this trade-off.

\begin{lemma}\label{lemma:program-size-double}
    Given an input program $\mathcal{P}$ and a conditional $cdt$ in it, and  let $\mathcal{P}^{\prime}$ be the program obtained by performing branch expansion by \cref{theorem:branch-expansion} on $cdt$, then $|\mathcal{P}^{\prime}| \leq 2\cdot|\mathcal{P}|$.  
\end{lemma}

\begin{theorem}
Given an input program $\mathcal{P}$, in which there are $c$ conditionals, and $depth\_limit$ is set to $d$, then the size of the output $\mathcal{P^{\prime}}$ of \cref{alg:optimization} is bounded by $\mathcal{O}(c\cdot2^{\frac{d(d+1)}{2}}\cdot|\mathcal{P}|)$. 
    \label{theorem:program-size-growth}
\end{theorem}
\begin{proof}
    Suppose conditionals $cdt_0, \dots, cdt_{d-1}$ are the $d$ conditionals in a program $\mathcal{P}_0$ that \cref{alg:optimization} expands at $d$ different nesting depths, where $\forall i\in\{1, \dots, d-1\}:$ $cdt_{i-1}$ is expanded before $cdt_{i}$. In the worst case, no two conditionals are nested within each other, and expanding $cdt_{i-1}$ makes both branches of $cdt_{i-1}$ expand to contain $cdt_{i}$. Then, at each nesting depth $k \in \{0, \dots, d\}$, the program contains $2^k$ copies of $cdt_k$. Suppose $\mathcal{P}_i$ is the program obtained after expanding every copy of $cdt_{i-1}$.
    Since, by \cref{lemma:program-size-double}, expanding a conditional at most doubles the size of the program, after expanding all copies of $cdt_k$, $|\mathcal{P}_{k+1}|$ is bounded by $2^k\cdot|\mathcal{P}_k|$. Hence $|\mathcal{P}_d|$ is bounded by $(\prod_{i=1}^{d}2^i)\cdot |\mathcal{P}_0| =2^{\sum_{i=1}^{d}i}\cdot|\mathcal{P}_0|=2^{\frac{d(d+1)}{2}}\cdot|\mathcal{P}_0|$.

    Since $\mathcal{P}$ has $c$ conditionals, the above procedure of expanding conditionals at $d$ different nesting depths is performed at most $\mathcal{O}(c)$ times. Therefore, the size of the output program is bounded by $\mathcal{O}(c\cdot2^{\frac{d(d+1)}{2}}\cdot|\mathcal{P}|)$.
\end{proof}

\begin{theorem}
Given an input program $\mathcal{P}$, where there are $c$ conditionals, and $depth\_limit$ is set to $d$, and assume the time complexity of performing the chosen circuit transpiler on $\mathcal{P}$ is $\mathcal{O}(T_{\mathcal{P}})$, then the time complexity of running the same circuit transpiler on $\mathcal{P^{\prime}}$, the output of \cref{alg:optimization} with input being $\mathcal{P^{\prime}}$, is $\mathcal{O}(T_{\mathcal{P}}\cdot c\cdot2^{\frac{d(d+1)}{2}}\cdot|\mathcal{P}|)$.
    \label{theorem:time-complexity}
\end{theorem}
\begin{proof}
    This follows directly \cref{theorem:program-size-growth}.
\end{proof}

\section{Evaluation}\label{sec:evaluation}
To establish the scope of our evaluation, we consider the stages of a quantum circuit's lifecycle in practical implementations. As illustrated by \cref{fig:circuit-lifecycle}, in fault-tolerant quantum computing (FTQC), circuits are typically encoded with quantum error correction (QEC) codes to protect against noise and errors. This encoding step transforms the circuit into a post-QEC circuit, where logical operations are performed on logical qubits encoded with multiple physical qubits. By contrast, pre-QEC circuits are the circuits prior to this encoding. Currently circuit optimizations are mainly performed on pre-QEC circuits, because QEC encoding transforms circuits into a much more complicated structure that is designed to counteract errors, making it more challenging to perform circuit-level optimization efficiently and effectively. Our method, however, could not only largely facilitate optimizations on pre-QEC dynamic circuits, as demonstrated in \cref{subsec:eval-pre-QEC}, but also could benefit optimizations on post-QEC circuits, as showcased in \cref{subsec:eval-post-QEC}.

\usetikzlibrary{shapes.geometric, arrows}
\tikzstyle{process} = [rectangle, rounded corners, minimum width=1cm, minimum height=0.8cm, text centered, draw=black, fill=blue!30]
\tikzstyle{arrow} = [thick,->,>=stealth]
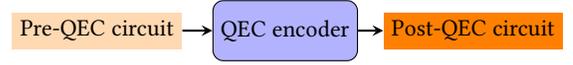
\begin{figure}[h]
    \centering
    \begin{tikzpicture}[node distance=2.5cm]

        \node [rectangle, draw=none, fill=orange!30](pre) {Pre-QEC circuit};
        \node (encoder) [process, right of=pre] {QEC encoder};
        \node (post) [rectangle, draw=none, fill=orange, right of=encoder] {Post-QEC circuit};

        \draw [arrow] (pre) -- (encoder);
        \draw [arrow] (encoder) -- (post);

    \end{tikzpicture}
    \caption{Stages of a circuit lifecycle in FTQC.}
    \label{fig:circuit-lifecycle}
\end{figure}

\subsection{Evaluation on Pre-QEC Circuits}\label{subsec:eval-pre-QEC}

\subsubsection{Experimental Setup}
To evaluate the performance of the proposed method, we conducted a series of experiments using randomly generated quantum circuits. The circuits are structured with a block-based template, as illustrated in \cref{fig:circuit-of-blocks}. In this template, each $B_i$ represents a $n$-qubit circuit block generated according to a specific pattern. By adjusting the number of blocks $k$, we control the size and scale of the circuits. Moreover, each block $B_{i+1}$ is made such that its expressions irreducibly depend (see \cref{def:irr-dependency}) on $B_i$, ensuring that the expansion process is well-defined and controlled.

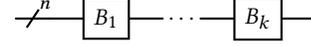
\begin{figure}
    \centering
    \begin{quantikz}[column sep = 0.45cm]
        &\qwbundle{n}&\gate{B_1}&\ \ldots\ &\gate{B_k}&
    \end{quantikz}
    \caption{Circuit template consisting of repetitive blocks of the same pattern. The number of blocks is $k$. All blocks are acting the same $n$-qubit set. }
    \label{fig:circuit-of-blocks}
\end{figure}

To analyze the performance of our method across a diverse range of circuit structures, we define two distinct block patterns:

\paragraph{Pattern 1: shallow conditionals}
The first pattern consists of blocks with shallow conditionals, where the nesting depth of each conditional is exactly 1. A block in this pattern is structured as shown in \cref{fig:circuit-pattern-one}. Here, $C_i$ is a sub-circuit that precedes the conditional branches, while $C_{a_i}$ and $C_{b_i}$ represent the sub-circuits executed in the "if" and "else" branches of the conditional, respectively. This pattern allows us to evaluate how our method performs in circuits with simple, non-nesting conditionals.
\begin{figure}
    \centering
    \begin{equation*}
        \begin{quantikz}
            &\qwbundle{n}&\gate{B_i}&
        \end{quantikz}
        \equiv
        \begin{quantikz}[column sep = 0.45cm]
        &\gate{C_i}&\gate{C_{a_i}}\gategroup[1,steps=1,style={inner
         sep=3.0pt}]{IF($E_i$)} &\gate{C_{b_i}}\gategroup[1,steps=1,style={inner
         sep=3.0pt}]{ELSE}&
        \end{quantikz}
    \end{equation*}
    \caption{Block pattern for circuits with shallow conditionals. $C_i$, $C_{a_i}$, $C_{b_i}$ are all randomly generated sub-circuits that are free of conditionals. The nesting depth of the block pattern is $1$. The circuit depth of $C_i$, $C_{a_i}$, $C_{b_i}$ are all fixed to be $d_s$.}
    \label{fig:circuit-pattern-one}
\end{figure}
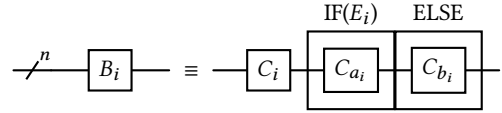

\paragraph{Pattern 2: recursively nesting conditionals} This pattern introduces recursively nesting conditionals. A block in this pattern is constructed recursively, as shown in \cref{fig:circuit-recursive-nested-pattern}. Each block $R_{i}[d]$ recursively contains blocks $R_i[d-1]$ within its branches. The nesting depth $d$ can be adjusted to control the depth of the conditionals in the circuit. One may notice that when $d = 1$, this pattern reduces to the pattern in \cref{fig:circuit-pattern-one}. This pattern enables us to evaluate the impact of our method on circuits with deeply nested conditionals.

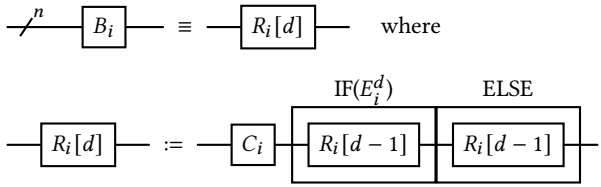
\begin{figure}
    \centering
    \begin{align*}
     &\begin{quantikz}
            &\qwbundle{n}&\gate{B_i}&
        \end{quantikz}
        \equiv
    \begin{quantikz}[column sep = 0.45cm]
        &\gate{R_{i}[d]}&
    \end{quantikz} \quad\text{where}
    \\
        &\begin{quantikz}[column sep = 0.45cm]
            &\gate{R_{i}[d]}&
        \end{quantikz}
        :=
        \begin{quantikz}[column sep = 0.45cm]
        &\gate{C_{i}}&\gate{R_i[d-1]}\gategroup[1,steps=1,style={inner
         sep=3.0pt}]{IF($E_i^{d}$)} &\gate{R_{i}[d-1]}\gategroup[1,steps=1,style={inner
         sep=3.0pt}]{ELSE}&
        \end{quantikz}
    \end{align*}

    \caption{Recursively nested pattern, where $d \ge 1$, $d$ is the nesting depth of the block pattern, $C_i$ is a random sub-circuit free of conditionals, and $R_i[0] = I$. The circuit depth of $C_i$ is fixed to be $d_s$.}
    \label{fig:circuit-recursive-nested-pattern}
\end{figure}

\subsubsection{Metrics}\label{subsubsec:metrics}
Before introducing the metrics used in our experiments, it is important to clarify why we should not use the whole circuit's gate count and circuit depth as metrics. While the gate count and depth of a circuit are widely accepted for evaluating the cost of the circuit, they fail to capture the structure of dynamic circuits with conditionals. Specifically, gate count and depth are typically calculated for the entire circuit, without accounting for the fact that, during execution, only one branch of a conditional is taken at a time. As a result, counting gates or depth across both branches simultaneously overestimates the cost of execution. Moreover, existing tools like Qiskit often treat a conditional block as a single gate or depth unit, completely disregarding the computational complexity of the branches within. This oversimplification makes such conventional metrics insufficient for evaluating the actual performance of dynamic circuits.

To address these limitations, we introduce the following metrics that are more aligned with the execution model of dynamic circuits. These metrics focus on evaluating the properties of execution paths (see \cref{def:execution-path}) within the circuit, such as the maximum and minimum depth or gate count across all paths. By tailoring our evaluation to the actual cost of executing specific paths, we provide a more accurate and meaningful comparison.

\begin{itemize}
    \item $max\text{-}p\text{-}depth$: The maximum number of layers of gates that must be executed sequentially along any execution path.
    \item $min\text{-}p\text{-}depth$: Similarly, the minimum depth along any execution path in the circuit.
    \item $max\text{-}p\text{-}gate\text{-}count$: The maximum number of gates along any execution path.
    \item $min\text{-}p\text{-}gate\text{-}count$: Similarly, the minimum gate count along any execution path in the circuit.
\end{itemize}

\begin{example}
    We use the circuit \cref{fig:ex-metrics} to demonstrate how the metrics are calculated. In the circuit, outside of the conditional, the gate count is $3$, and the depth is $2$. In the if-branch, the gate count is $1$, and the depth is $1$; in the else-branch, the gate count is $3$, and the depth is $2$.
    Then, in case the if-branch is executed, the total depth of the circuit is $2 + 1 = 3$, and the total gate count is $3 + 1 = 4$; in case the else-branch is executed, the total depth is $2 + 2 = 4$, and the total gate count is $3 + 3 = 6$. Therefore, $max\text{-}p\text{-}depth = max(3, 4) = 4$, $min\text{-}p\text{-}depth = min(3, 4) = 3$,
    $max\text{-}p\text{-}gate\text{-}count = max(4, 6) = 6$,
    $min\text{-}p\text{-}gate\text{-}count = min(4, 6) = 4$.

\begin{figure}
    \centering
    \begin{quantikz}
        & \gate{X}&\gate{Y}&\gate{H}\gategroup[2,steps=1,style={inner sep=3.7pt}]{IF($E$)}&\gate{H}\gategroup[2,steps=2,style={inner sep=3.7pt}]{ELSE}&\gate{X}&\\
        & \gate{H}&        &      &\gate{Z}&&     
    \end{quantikz}  
    \caption{An example $2$-qubit circuit with one conditional, where $E$ is an expression. This circuit is used to demonstrate how metrics are calculated.}
    \label{fig:ex-metrics}
\end{figure}
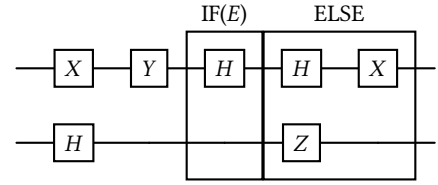
\end{example}

Intuitively, a lower value of these metrics indicates a more efficient circuit. By comparing these metrics before and after applying our method, we can evaluate its impact on circuit optimization.

\subsubsection{Experimental Procedure}
We set the qubit number $n = 3$ (see \cref{fig:circuit-of-blocks}) and set $d_s = 5$ (see \cref{fig:circuit-pattern-one} and \cref{fig:circuit-recursive-nested-pattern}). For each block pattern, we vary the number of blocks $k$ from $1$ to $100$. For each value of $k$, we generate $100$ random circuits based on the current block pattern. 
Then, the experiments are conducted under two settings:

\begin{itemize} 
    \item \textbf{Baseline optimization:} Apply Qiskit's transpiler with optimization level $3$ directly to the circuits and collect the metrics ($max\text{-}/min\text{-}p\text{-}depth/\text{-}gate\text{-}count$, see \cref{subsubsec:metrics}) for each of the $100$ random circuits. 
    \item \textbf{Preprocessed optimization:} First, our procedure \cref{alg:optimization} is applied to the circuits. Then the same transpiler with the same optimization level runs on these preprocessed circuits. Metrics are then collected for each of the $100$ random circuits. For circuits generated according to the pattern \cref{fig:circuit-pattern-one}, the parameter $depth\_limit$ in \cref{alg:optimization} is fixed at $1$ during preprocessing. However, for circuits generated according to the pattern \cref{fig:circuit-recursive-nested-pattern}, our method is evaluated under varying $depth\_limit$ values of $1$, $2$, and $3$. Metrics are collected separately for each $depth\_limit$ setting.

\end{itemize}
For both settings, the four metrics are calculated for each of the $100$ random circuits generated at the same set of parameters ($n$, $d_s$, $k$, and circuit pattern) and then averaged to provide a final result for each experiment.
By comparing the metrics for optimized circuits in these two settings, we can quantify the benefits of our method as a preprocessing step.
The random generation ensures diversity in circuit structure. The metrics we calculate are averaged across circuits to reduce the effect of outliers and provide robust results.

\subsubsection{Results}
\paragraph{Performance on pattern $1$}
\Cref{fig:perc-dec-depth} and \cref{fig:perc-dec-gate-count} present data collected from quantum circuits generated according to the block pattern \cref{fig:circuit-pattern-one}, i.e., circuits with shallow conditionals. These plots show the percentage reduction in the four metrics ($max\text{-}p\text{-}depth$, $min\text{-}p\text{-}depth$, $max\text{-}p\text{-}gate\text{-}count$, and $min\text{-}p\text{-}gate\text{-}count$) when our preprocessing method is applied before running Qiskit's transpiler, compared to running the transpiler directly. From \cref{fig:perc-dec-depth} and \cref{fig:perc-dec-gate-count}, the following analytic results are drawn:
\begin{itemize}
    \item \textbf{Effectiveness across metrics:} The consistent reduction in all metrics suggests that our method is effective in exposing optimization opportunities for Qiskit's transpiler.
    \item \textbf{Scalability:} The stabilization of percentage reductions across all metrics as the circuit size (i.e., the block number) increases demonstrates the scalability of our method.
    \item \textbf{Improvement in diverse execution paths:} Reductions in both $max\text{-}p$ and $min\text{-}p$ metrics imply that our method is beneficial not only for optimizations on the "best-case" execution paths but also for improving optimizations on the "worst-case" paths. 
\end{itemize}

\begin{figure}
\centering
\begin{tikzpicture}[>=latex,node distance=1.5em]
 \node(a)
 {
\scalebox{.6}{
    \input{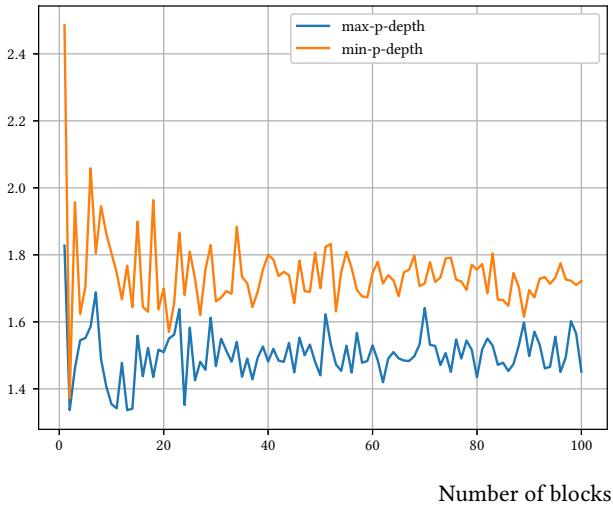}
}
 };

 \node[] at (-3.0,3.6)
 {\text{Percentage decrease}};
 \node[] at (2.9,-3.5) {\text{Number of blocks}};
\end{tikzpicture}
    \caption{The percentage decrease in metrics on depth
    after applying our method on circuits of pattern \cref{fig:circuit-pattern-one}.}
    \label{fig:perc-dec-depth}
\end{figure}

\begin{figure}
\centering
\begin{tikzpicture}[>=latex,node distance=1.5em]
 \node(a)
 {
\scalebox{.6}{
    \input{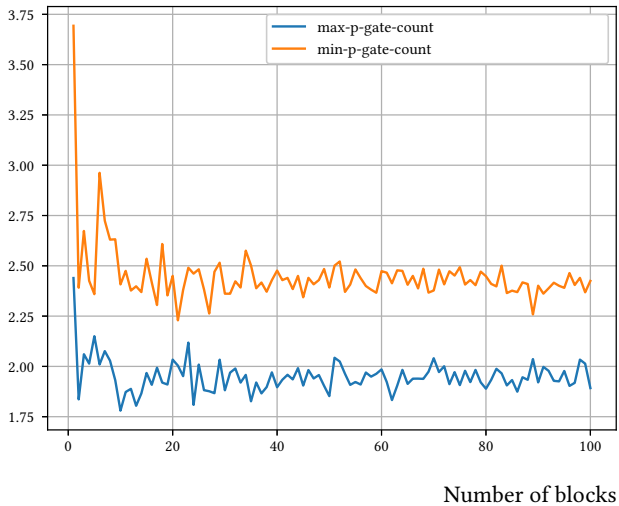}
}
 };
 \node[] at (-3.0,3.6)
 {\text{Percentage decrease}};
 \node[] at (2.9,-3.5) {\text{Number of blocks}};
\end{tikzpicture}
    \caption{The percentage decrease in metrics on gate count
    after applying our method on circuits of pattern \cref{fig:circuit-pattern-one}.}
    \label{fig:perc-dec-gate-count}
\end{figure}

\paragraph{Performance on pattern $2$}
\Cref{fig:perc-dec-depth-pattern-2} and \cref{fig:perc-dec-gate-count-pattern-2} present data collected from circuits generated according to the block pattern \cref{fig:circuit-recursive-nested-pattern}, i.e., circuits with relatively deep nesting conditionals, where we set the nesting depth $d = 4$. These plots show the percentage reduction in the four metrics when our preprocessing method is applied under different exploration depth $depth\_limit$ before running Qiskit's transpiler, compared to running the transpiler directly. From \cref{fig:perc-dec-depth-pattern-2} and \cref{fig:perc-dec-gate-count-pattern-2}, the following results are drawn:
\begin{itemize}
    \item \textbf{Effectiveness and scalability:} The performance of our method on circuits generated by pattern \cref{fig:circuit-recursive-nested-pattern} reconfirms its ability to consistently expose optimization opportunities and maintain stable improvements across all metrics, even as circuit size increases. This aligns with the conclusions drawn from its performance on circuits of pattern \cref{fig:circuit-pattern-one}.
    \item \textbf{Influence of $depth\_limit$:} Since increasing $depth\_limit$ improves the percentage decrease for all metrics, with $depth\_$ $limit = 3$ offering the best results, we conclude that our method is better at exposing optimization opportunities when running at a deeper exploration depth $depth\_limit$. 
\end{itemize}

\begin{figure}
\centering
\begin{tikzpicture}[>=latex,node distance=1.5em]
 \node(a)
 {
\scalebox{.6}{
    \input{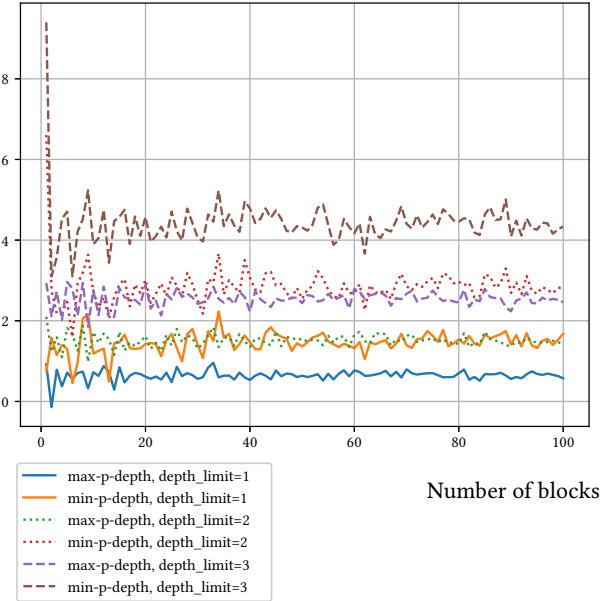}
}
 };

 \node[] at (-3.0,4.3)
 {\text{Percentage decrease}};
 \node[] at (2.9,-2.5) {\text{Number of blocks}};
\end{tikzpicture}
    \caption{The percentage decrease in 
    metrics on depth
    after applying our method on circuits of pattern \cref{fig:circuit-recursive-nested-pattern}.}
    \label{fig:perc-dec-depth-pattern-2}
\end{figure}

\begin{figure}
\centering
\begin{tikzpicture}[>=latex,node distance=1.5em]
 \node(a)
 {
\scalebox{.6}{
    \input{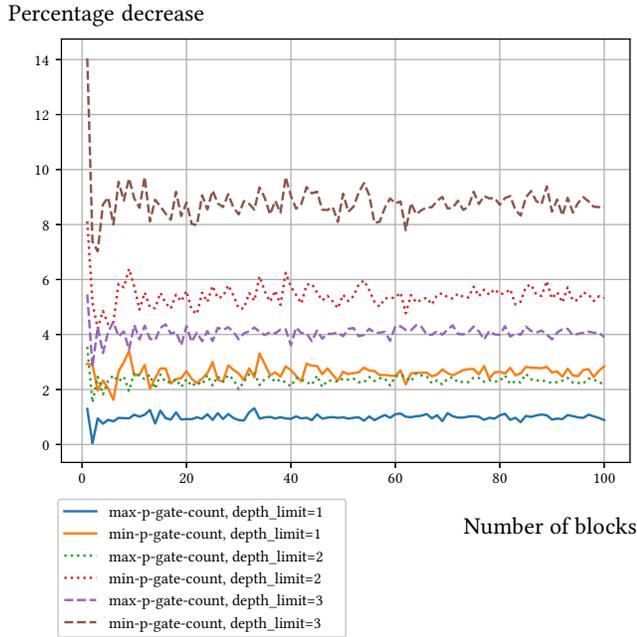}
}
 };

 \node[] at (-3.0,4.3)
 {\text{Percentage decrease}};
 \node[] at (2.9,-2.5) {\text{Number of blocks}};
\end{tikzpicture}
    \caption{The percentage decrease in 
    metrics on gate count
    after applying our method on circuits of pattern \cref{fig:circuit-recursive-nested-pattern}.}
    \label{fig:perc-dec-gate-count-pattern-2}
\end{figure}

\paragraph{Discussion}
The results across block patterns demonstrate the consistent effectiveness and scalability of our method in facilitating circuit optimization. For both patterns \cref{fig:circuit-pattern-one} and \cref{fig:circuit-recursive-nested-pattern} circuits, our method achieves stable reductions in execution path depth and gate count, showcasing its robustness even as circuit size increases. For pattern \cref{fig:circuit-recursive-nested-pattern} circuits, the $depth\_limit$ parameter in our procedure allows us to control the trade-off between the strength of our exploration and the computational overhead, with deeper exploration yielding more spots for optimization. The consistent performance improvements across diverse circuit structures demonstrate our method is adaptable to different circuit patterns. Collectively, these findings highlight the potential of our approach to enhance the optimizations of existing transpilation frameworks, such as Qiskit’s, by systematically uncovering hidden optimization opportunities.

\subsection{Optimization in Fault-Tolerant Circuits}\label{subsec:eval-post-QEC}

While there have been many works on optimizing encoding schemes or the physical implementation of QEC-encoders to produce more efficient fault-tolerant circuits \cite{nautrup2019optimizing, burkard1999physical}, optimization on circuits after they are encoded by a QEC scheme remains relatively unexplored.
Our method, however, easily presents opportunities to improve post-QEC circuits by uncovering branch-specific simplifications. 

\paragraph{Construction of post-QEC circuits}
The construction of a circuit encoded with any QEC scheme is also a sequence of repetitive blocks acting on the same $n$-qubit system, as shown in \cref{fig:circuit-of-blocks}. 
Each block $B_i$ comprises the following components, as depicted in \cref{fig:QEC-circuit}:
\begin{itemize}
    \item Logical Gate Application ($L_i$): Implements a logical operation on the encoded qubits.
    \item Syndrome Extraction ($S_i$): Performs measurements on ancilla qubits to reveal a syndrome pattern associated with specific errors that might have occurred in the system. These measurements extract information about the error without collapsing the encoded logical state.
    \item Error Correction ($C_i$): Based on the syndrome information extracted, the error correction operation counteracts the detected error by applying a corresponding corrective transformation to the system.
\end{itemize}
This modular construction allows the circuit to mitigate errors while maintaining the logical operations required for computation, making it a general framework for FTQC.

\begin{figure}
    \centering
    \begin{align*}
    \begin{quantikz}
            &\qwbundle{n}&\gate{B_i}&
    \end{quantikz}
    \equiv 
    \begin{quantikz}
        &\qwbundle{n}&\gate{L_i}&\gate{S_i}\vcwdouble{1-4}{2-4}{}&\gate{C_i}\wire[d][1]{c}& \\ 
        \setwiretype{c}
        &&&&\phase{}& 
    \end{quantikz}
    \end{align*}
    \caption{The block pattern for a QEC-encoded circuit.}
    \label{fig:QEC-circuit}
\end{figure}
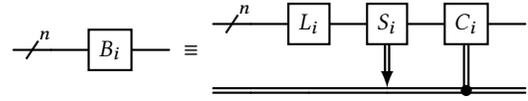

\paragraph{Potential of optimization revealed by our method}
In a post-QEC circuit, the error correction operation $C_i$ can be interpreted as a set of conditionals, of which each branch corresponds to a specific syndrome pattern. In each branch, a specific corrective action is applied to counteract the identified error. By expanding these conditionals, our method reveals new optimization opportunities that are otherwise obscured in the original program structure. Notably, the logical operation $L_{i+1}$ in the subsequent block is independent of $C_i$, meaning that when expanding the conditionals within $C_i$, $L_{i+1}$ will be included in the expanded branches. This enables branch-specific optimizations where $C_i$ and  $L_{i+1}$ may combine in a way that reduces redundant operations. See the following example.
\begin{example}
    Consider Shor's $9$-qubit code. If, in the block $B_i$, the corrective action determined by $C_i$ is to apply an $X$ gate on the $j$-th qubit, and the subsequent logical operation $L_{i+1}$ is a logical $Z$ gate, which is supposedly implemented by applying $X$ gates to all $9$ qubits, then the $X$ gates on the $j$-th qubit cancel out.
\end{example}

\section{Related works}\label{sec:rel_works}
Dynamic circuits optimization has been a critical area of research \cite{Karuppasamy_2025}, with significant contributions from state-of-the-art toolchains such as Qiskit, T$\ket{\text{Ket}}$, and PyZX \cite{Qiskit, Sivarajah_tket_2021, EPTCS318.14}. These frameworks provide various transpilation techniques to reduce circuit depth, gate count, and hardware constraints, enabling efficient execution on quantum devices. However, their focus has predominantly been on static circuits, with limited attention given to the unique optimization challenges posed by dynamic circuits. Recent efforts have started to explore optimization techniques for hybrid quantum-classical workflows and circuits with measurements \cite{10821341, Borgna_2021, fang2023dynamicquantumcircuitcompilation}. Other tasks, such as 
 circuit verification, when involving dynamic circuits, pose greater challenges compared to static circuits, which makes the study of dynamic circuits an active research field \cite{burgholzer2021towards}.

\section{Conclusion}\label{sec:conclusion}
We presented a method to enhance the optimization of dynamic quantum circuits by expanding nested conditionals, exposing opportunities that are otherwise hidden. 
Experimental results demonstrated consistent reductions in gate count and depth in execution paths across diverse circuit patterns and also validated the practical benefits of the trade-off between the optimization potential and the runtime overhead controlled by $depth\_limit$. 
Its capability to facilitate optimizations in error-corrected circuits is also showcased.

\begin{acks}
I thank the reviewers for their efforts and their helpful suggestions. I am grateful to my supervisor Prof. Dr. Helmut Seidl for many fruitful discussions and his support at all times.
The research is part of the Munich Quantum Valley (MQV), which is supported by the Bavarian state government with funds from the Hightech Agenda Bayern Plus. 
\end{acks}

\bibliographystyle{ACM-Reference-Format}

\end{document}